\documentclass[pra,twocolumn,preprintnumbers,amsmath,amssymb,superscriptaddress]{revtex4}

\usepackage{graphicx}
\usepackage{dcolumn}
\usepackage{bm,color}
\usepackage{amsmath, latexsym, amssymb, amscd, amsfonts, amsthm, epsfig, mathrsfs, graphicx, url,verbatim}
\usepackage[english]{babel}
\usepackage[braket]{qcircuit}

\newtheorem{theorem}{Theorem}

\newtheorem{definition}{Definition}

\newtheoremstyle{mytheo}
  {16pt}
  {0pt}
  {}
  {}
  {\bfseries}
  {:}
  {.5em}
  {}

\theoremstyle{mytheo}

\begin{document}

\title{Braid map definition of $SLH$-networks and the series product}

\author{Luc Bouten}


\begin{abstract}
This article investigates a cascade of quantum systems within an 
$SLH$-network. The original definition of an $SLH$-network depends 
on an open conjecture regarding the essential selfadjointness of 
the generator of the time evolution of the network \cite[Conjecture 11]{GoJ09b}. This problem 
is avoided by using  the so-called braid map, which leads 
to a rigorous definition of the $SLH$-network \cite{BoG26}. 
An expression is derived for the 
time-evolution of the network with finite delay.
The series product is recovered in the instantaneous 
feed-forward limit as a theorem that no longer depends on 
the conjecture \cite[Conjecture 11]{GoJ09b}.
\end{abstract}

\maketitle

\section{Introduction}

The initial mathematical definition of the 
time evolution of an 
$SLH$-network was given by a symmetric 
operator which was 
conjectured to be essentially self-adjoint 
on its given domain \cite[Conjecture 11]{GoJ09b}. 
Since the conjecture is still open, the definition 
of the network is conditional on the conjectured 
essential self-adjointness of the generator of 
its time evolution. There are three possible 
outcomes for the conjecture: 1.\ the conjecture 
is true, in which case there is no problem, 2.\ the 
conjecture is false and there exist many self-adjoint 
extensions, in which case the definition is not unique, 3.\
the conjecture is false and there exist no self-adjoint extensions 
in which case the network cannot be defined in this way. 
Since it is not clear how the conjecture will be resolved, it 
remains ambiguous whether the Gough-James definition 
will lead to a sound mathematical object. 

This problem with the essential self-adjointness of 
the generator of the time evolution is solved 
in \cite{BoG26} by resorting to an alternative definition
of an $SLH$-network. The network evolution is 
defined by the unitary cocycle that solves the quantum 
stochastic differential equation (QSDE) \cite{HuP84, GC84} of the network without 
the connections. This time evolution 
is then periodically interrupted, swapping 
 the output and input of channels that are connected. This is 
done by  the so-called \emph{braid map} \cite{BoG26}. 
This procedure provides an unambiguous and rigorous definition of an 
$SLH$-network \cite{BoG26}.  
 
The most important feature of the $SLH$-framework is the set of 
algebraic rules for composing and reducing networks 
\cite{GoJ09, GoJ09b, CKS17}. One of the most important 
algebraic relations is the so-called \emph{series product}
\cite{Gar93, Car93a, GoJ09, GoJ09b}. 
The series product provides the $SLH$-triplet for a network 
in which the $m$ outputs of one node are fed forward to 
the $m$ inputs of another node in the zero delay limit.
Since \cite{BoG26} provides a rigorous definition of an $SLH$-network, 
it is now possible to derive the series product using the braid map 
definition of the network. In essence, this transforms the series product 
from a formal rule (rule conditional on the resolution of the conjecture) 
to a theorem.

In this article it is shown that in the instantaneous feed-forward 
(zero delay) limit, the network time evolution converges strongly, 
uniformly on compact time intervals, to the solution of the 
QSDE associated to the $SLH$-triple given by the series product, 
see equations \eqref{equation main result} and \eqref{equation series product}  below. Along the 
way an expression is obtained for the network evolution 
with finite delay, see equation \eqref{equation second result} below.

\section{The quantum stochastic calculus}

A single node in an $SLH$-network consists of a quantum 
system represented on some separable 
complex Hilbert space $\mathcal{H}$ 
and $m$ input-output channels. 
We call $m$ the multiplicity of 
the node.

For an interval $I\subset \mathbb{R}$, we 
define the \emph{bosonic Fock space with multiplicity} $m$ 
as
\begin{equation}\label{def Fock space}
\mathcal{F}_I = \mathbb{C} \oplus \bigoplus_{n =1}^\infty L^2(I; \mathbb{C}^m)^{\otimes_s n}.
\end{equation}
Here $L^2(I; \mathbb{C}^m)$ stands for 
the Hilbert space of  all quadratically integrable 
functions on $I$ with values 
in $\mathbb{C}^m$. 
Note that the Fock space $\mathcal{F}_I$ 
describes a field of bosons, which we will 
always take to be a field of photons. 
The $m$ input-output channels are 
represented on the symmetric Fock 
space $\mathcal{F}_\mathbb{R}$ and 
should be interpreted as $m$ channels 
in the electromagnetic field. We will 
often shorten $\mathcal{F}_\mathbb{R}$
to $\mathcal{F}$.
The complete node is then represented 
on the Hilbert space $\mathcal{H}\otimes\mathcal{F}$.

For $f \in L^2(I; \mathbb{C}^m)$, we define 
the \emph{exponential vector} $e(f) \in \mathcal{F}_I$ as
\begin{equation}\label{def exponential vector}
e(f) = 1 \oplus \bigoplus_{n= 1}^\infty \frac{1}{\sqrt{n!}}f^{\otimes n}.
\end{equation}
We call the linear span of the exponential vectors 
the \emph{exponential domain} 
which is dense in $\mathcal{F}_I$.

We  introduce 
for all  $1\le i,j \le m$ and $t\ge 0$
the fundamental noises $A^i_t,\ A^{i*}_t$ 
and $\Lambda^{ij}_t$ 
on the exponential domain of $\mathcal{F}$
by (see also \cite{HuP84, Par92})
  \begin{equation}\label{def fundamental noises}\begin{split}
  &A^i_t e(f) = \left(\int_0^t f_i(s)ds\right)\, e(f),\\
  &\big\langle e(g), A^{i*}_t e(f)\big\rangle  = 
  \left(\int_0^t\overline{g}_i(s)ds\right)\, \big\langle e(g),
  e(f)\big\rangle, \\
  &\big\langle e(g), \Lambda^{ij}_t e(f)\big\rangle  = 
  \left(\int_0^t\overline{g}_i(s)f_j(s)ds\right)\, 
  \big\langle e(g),e(f)\big\rangle, 
  \end{split}\end{equation}
for all $f$ and $g$ in $L^2(\mathbb{R}; \mathbb{C}^m)$.
$A^i_t$ and $A^{i*}_t$ are called the 
\emph{annihilation} and \emph{creation}
processes, respectively. The processes 
$\Lambda^{ij}_t$ are called \emph{gauge} 
processes.

For an interval $I \subset \mathbb{R}$, we denote by 
$\chi_I$ the \emph{indicator function of} $I$, i.e.\ the 
function that is 1 on I and 0 elsewhere in $\mathbb{R}$.
For all $f \in L^2(\mathbb{R}; \mathbb{C}^m)$ and
$I \subset \mathbb{R}$, we can now define $f_I$ by
\begin{equation*}
f_{I}(x) = f(x)\chi_I(x), \ \ \mbox{for all}\ \ x \in \mathbb{R}.
\end{equation*}
We will use the shorthand $f_{t)} = f_{(-\infty, t)}$ and 
$f_{[t} = f_{[t, \infty)}$ in the following. There exists 
a unique unitary isomorphism $\iota$ from 
$\mathcal{F}$ to $\mathcal{F}_{(-\infty,t)}\otimes\mathcal{F}_{[t,\infty)}$ 
such that (see e.g.\ \cite[Prop 19.6]{Par92}) 
\begin{equation*}
\iota\big(e(f)\big) = e(f_{t)})\otimes e(f_{[t}), 
\end{equation*}
 for all $f \in L^2(\mathbb{R}; \mathbb{C}^m)$ and $t \ge 0$.
In the following we 
will simply identify $\mathcal{F}$ and 
$\mathcal{F}_{(-\infty,t)}\otimes\mathcal{F}_{[t,\infty)}$, 
i.e.\ to keep notation light, we will no 
longer write the unitary map $\iota$ between them.
A process $L_s,\, (s\ge 0)$ on $\mathcal{H}\otimes \mathcal{F}$ 
is called \emph{adapted} if $L_s$ acts nontrivially only 
on $\mathcal{H}\otimes \mathcal{F}_{(-\infty, s)}$ and 
acts as the identity on $\mathcal{F}_{[s, \infty)}$ for 
all $s\ge 0$.

It is possible \cite{HuP84} to define  stochastic 
integrals of adapted processes $L_s$ against 
the fundamental noises. That is, it is possible to give meaning 
to the expression $X_t = X_0 + \int_0^tL_sdM_s$ where 
$M_s$ is one of the fundamental noises $A^i_s,\ A^{i*}_s$ 
or $\Lambda^{ij}_s$. The expression can be written in 
shorthand as $dX_t = L_tdM_t$. 

More importantly, these stochastic integrals
satisfy a calculus with which they can be 
manipulated in calculations \cite{HuP84}. 
The calculus consists of the following. Suppose $X_t$ and 
$Y_t$ are stochastic integrals, i.e.\ 
$dX_t = L^1_tdM^1_t$ and $dY_t = L^2_tdM^2_t$ where 
$L^1$ and $L^2$ are adapted processes and $M^1$ 
and $M^2$ are fundamental noises, then the product 
$X_tY_t$ is itself a stochastic integral. Moreover, 
the product $X_tY_t$ satisfies the following 
quantum It\^o rule, (integration by parts rule)
  \begin{equation}\label{equation integration by parts}
  d(X_t Y_t) = X_tdY_t + (dX_t)Y_t + dX_tdY_t, 
  \end{equation}
where to evaluate $dX_tdY_t$ we use that the 
increment $dM_t$ of a fundamental noise commutes 
with all adapted processes, and products $dM^1_tdM^2_t$ 
are given by the following quantum It\^o table \cite{HuP84}  
\begin{center}
{\large \begin{tabular} {l|lll}
$dM^1\backslash dM^2$ & $dA^{i*}_t$ & $d\Lambda^{ij}_t$ & $dA^i_t$ \\
\hline 
$dA^{k*}_t$ & $0$ & $0$ & $0$ \\
$d\Lambda^{kl}_t$ & $\delta_{li}dA^{k*}_t$ & $\delta_{li}d\Lambda^{kj}_t$ & $0$  \\
$dA^k_t$ & $\delta_{ki}dt$ & $\delta_{ki}dA^j_t$ & $0$ 
\end{tabular} }
\end{center}
and all products $dM_tdt$ and $dtdM_t$ are zero.
As an example, suppose $dX_t = L^1_tdA^i_t$ and 
$dY_t = L^2_tdA^{i*}_t$, then 
$d(X_tY_t) = X_tL^2_tdA^{i*}_t + L^1_tY_tdA^{i}_t + L^1_tL^2_tdt$. 

The definition of the quantum stochastic integral 
and its associated calculus \cite{HuP84} also immediately 
lead to the study of quantum stochastic differential equations (QSDE's). 
In this paper we study the following QSDE on $\mathcal{H}\otimes \mathcal{F}$
\begin{equation}\label{definition Ut}\begin{split}
    dU_t& = \Bigg\{\sum_{i,j = 1}^m\big(S_{i j}-\delta_{ij}\big) d \Lambda^{ij}_t + \sum_{i=1}^m L_i dA^{i*}_{t}\\ 
     & -\sum_{i,j = 1}^m L^*_i S_{i j}dA^j_t 
    -\frac{1}{2}\sum_{i= 1}^m L^*_iL_i dt - iHdt\Bigg\}U_t, 
\end{split}\end{equation} 
with initial condition $U_0 =I$ and $t\ge 0$. Here $(S, L, H)$ is 
a so-called $SLH$-triple, i.e.\ $S$ is a collection 
of bounded operators $S_{ij}, \ 1 \le i,j \le m$ on $\mathcal{H}$, 
such that
\begin{equation*}
  \sum_{j=1}^m S_{ji}^*S_{jl} = \delta_{il}, \ \ \ \ \sum_{j=1}^m S_{ij}S_{lj}^* = \delta_{il},
\end{equation*}
(i.e.\ $S$ is a unitary $m\times m$ matrix with coefficients in $\mathcal{H}$), 
$L$ is a collection of bounded operators $L_i, 1\le i\le m$ on 
$\mathcal{H}$ and $H$ is a self-adjoint bounded operator 
on $\mathcal{H}$.

It follows from \cite{HuP84} that an adapted 
process $U_t$ that solves
Eqn.\ \eqref{definition Ut} exists, is 
unique, and is strongly continuous in $t$. Furthermore, it follows from \cite{HuP84}
that the solution to Eqn.\ \eqref{definition Ut} is 
unitary if and only if the coefficients of the 
equation are as in Eqn.\ \eqref{definition Ut} 
and form an $SLH$-triple. 

We are now ready to introduce the time 
evolution of a single node of a network. We 
let $\theta_t$ be the \emph{left shift} on 
$L^2(\mathbb{R}; \mathbb{C}^m)$, i.e.\
\begin{equation*}
\theta_t(f)(x) = f(x+t)\ \ \ \  \mbox{for all}\ x \in \mathbb{R}.
\end{equation*}
We denote by $\Theta_t: \mathcal{F} \to \mathcal{F}$ 
the second quantisation of the map $\theta_t$. $\Theta_t$ represents the 
free time evolution of the $m$ input-output channels. The photons simply
fly from right to left in units such that $c$, the speed of light, 
equals $1$ which allows us to identify time $t$ with position $x$ 
along the channel.

The quantum system with Hilbert space $\mathcal{H}$ is 
located at the origin of the node and perturbs the free 
evolution $\Theta_t$ in a way that is completely determined 
by the $SLH$-triple associated to the node. The $SLH$-triple 
of the node determines via Eqn.\ \eqref{definition Ut}
an adapted unitary process $U_t$ that is a 
cocycle with respect to the shift $\Theta_t$, i.e.
\begin{equation}\label{cocycle Ut}
U_{t+s} = \Theta_{-s}U_t\Theta_{s} U_s.
\end{equation}
Note that $U_t$ is not a one-parameter group 
of unitaries. We can however, define a one-parameter 
group of unitaries $\hat{U}_t$ by
\begin{equation}\label{definition Uhat}
\hat{U}_t =  \left\{ \begin{array}{ll}
 \Theta_t U_t & \mbox{\ \ \ if \ } t \ge 0 \\
 U^*_{-t} \Theta_t & \mbox{\ \ \ if \ } t < 0 
  \end{array}\right. .
\end{equation}

\section{$SLH$-networks}\label{section SLH-networks}

Suppose that we have $N$ separate components with 
initial spaces $\mathcal{H}_1,\ldots, \mathcal{H}_N$. 
We next introduce channels with multiplicities 
$m_1,\ldots, m_N$, symmetric Fock spaces 
$\mathcal{F}_{1},\ldots, \mathcal{F}_{N}$ and $SLH$-triples $(S^1, L^1, H^1),\ldots, (S^N, L^N, H^N)$. 
We can now introduce:
\begin{equation*}\begin{split}
&\mathcal{H} = \mathcal{H}_1\otimes\ldots\otimes\mathcal{H}_N, \\
&m = m_1 + \ldots+m_N, \\
&\mathcal{F} = \mathcal{F}_1 \otimes\ldots\otimes\mathcal{F}_N.
\end{split}\end{equation*}
As $\mathcal{F}_i$ is the symmetric second-quantisation of 
$L^2(\mathbb{R}; \mathbb{C}^{m_i})$, we can again 
use \cite[Prop 19.6]{Par92} to exploit the canonical 
isometry between the tensor products of the Fock spaces 
and the symmetric second-quantisation of the direct sum 
of the initial spaces to conclude that
\begin{equation*}
\mathcal{F} = \mathbb{C}\oplus \bigoplus_{n=1}^\infty L^2(\mathbb{R}; \mathbb{C}^m)^{\otimes_s n}.
\end{equation*} 
That is, the combined system of all the components 
has Hilbert space $\mathcal{H}\otimes\mathcal{F}$, 
where $\mathcal{F}$ has multiplicity $m$. By ampliation 
with the identity, we can now extend all operators in 
the $SLH$-triples of the components and all the 
quantum noises to the space $\mathcal{H}\otimes\mathcal{F}$. 
We re-label the noises such that the labels run from $1$ to $m$. 
That is, the noises labeled with $1,\ldots,m_1$ are the 
noises that act non-trivially on $\mathcal{F}_1$, 
the noises labeled with $m_1+1, \ldots,m_1+m_2$ 
are the noises that act non-trivially on $\mathcal{F}_2$ 
and so on up to the noises labeled by 
$(\sum_{i=1}^{N-1}m_i) +1,\ldots,m$ 
which are the noises that act non-trivially on $\mathcal{F}_N$.

\begin{definition} {\bf \cite{GoJ09, GoJ09b}}\label{definition SLH concatenation}
The $SLH$-triple $(S,L,H)$ of the network that consists of 
$N$ components with initial spaces $\mathcal{H}_1,\ldots, \mathcal{H}_N$, 
multiplicities $m_1,\ldots, m_N$, 
symmetric Fock spaces $\mathcal{F}_{1},\ldots, \mathcal{F}_{N}$ 
and $SLH$-triples $(S^1, L^1, H^1),\ldots, (S^N, L^N, H^N)$, is given by 
\begin{equation*}\begin{split}
&S =
\begin{bmatrix}
S_{11} & \ldots & S_{1m} \\
\vdots & & \vdots\\
S_{m1}& \ldots & S_{mm}
\end{bmatrix} = 
\begin{bmatrix}
S^1 & 0 & \ldots & 0 \\
0 & S^2 & \ldots & 0 \\
\vdots & \vdots & & \vdots \\
0 & 0&\ldots&S^N
\end{bmatrix},                      \\
&L = \begin{bmatrix} 
L_1 \\
\vdots\\
L_m
\end{bmatrix}
= 
\begin{bmatrix}
L^1\\
\vdots \\
L^N
\end{bmatrix},\ \ \
H = \sum_{i=1}^N H^i.
\end{split}\end{equation*} 
We denote this as 
\begin{equation*}
(S, L, H) = (S^1, L^1, H^1) \boxplus (S^2, L^2, H^2) \boxplus \ldots \boxplus (S^N, L^N, H^N).
\end{equation*}
\end{definition}

The time evolution of the network, 
consisting of the $N$ components (without 
connections between inputs and outputs which 
will be introduced below), is governed by 
the cocycle $U_t$ which satisfies equation  
\eqref{definition Ut} where the $SLH$-triple $(S,L,H)$ 
is given by Definition \ref{definition SLH concatenation}. 
It follows that $U_t = U_t^1 \otimes \ldots \otimes U_t^N$. 

By connecting specific outputs and inputs, a genuine 
network emerges.  An output can be 
connected only to a single input, and \textit{vice versa}. Note 
that it is possible to connect the output of a certain channel 
to the input of that same channel.
 
Each feedback connection $\mathsf{c}$ is a triple $(r, s, \xi)$ 
of an input $r$ (\textit{range}) and an output 
$s$ (\textit{source}), i.e., $s, r \in \{1,\ldots,m\}$. 
So the output $s$ is fed back into the network as 
input to $r$ and the corresponding time delay is 
$\xi  >0$.  $\mathcal{C}$ denotes the set of all 
connections in the system. The shortest line length 
in the system is denoted by $\xi_{\text{min}} =
\mbox{min}\big\{ \xi ; \ (s,r,\xi) \in \mathcal{C}\big\}$.

\subsection{The Braid Map}

The braid map will be defined as a unitary: for a certain 
connection, it swaps what is present in 
the output with what is present in 
the input.  For a certain connection $\mathsf{c} = (r,s,\xi)$ there 
are $4$ mutually exclusive possibilities 
for the $i$th channel:
\begin{enumerate}
\item Channel $i$ is neither the output (source) nor the input (range) of the connection $\mathsf{c}$: $i \neq r$ and $i \neq s$,
\item The connection $\mathsf{c}$ is such that $s \neq r$ and channel $i$ is the input (range), i.e.\ $i =r$,
\item The connection $\mathsf{c}$ is such that $s \neq r$ and channel $i$ is the output (source), i.e.\ $i =s$,
\item The connection $\mathsf{c}$ is such that $s = r$ and channel $i$ is the output (source) and the input (range), i.e.\ $i =s = r$.
\end{enumerate}

\begin{definition} {\bf\cite{BoG26}}\label{definition braid map}
Let $0 < \sigma \le \xi_{\text{min}}$. 
We define a unitary map 
$b_\sigma(\mathsf{c}):\ L^2(\mathbb{R}; \mathbb{C}^m) \to L^2(\mathbb{R}; \mathbb{C}^m)$ 
by defining its vector components 
$\big(b_\sigma (\mathsf{c})f\big)_i$ for $i = 1,\dots, m$.
The $i$th vector component is defined 
according to which of the four possibilities the channel $i$ 
belongs to for the given connection $\mathsf{c}$:
\begin{enumerate}
\item if $i \neq r$ and $i \neq s:\ \forall t \in \mathbb{R}:$ 
\begin{equation*}
\big(b_\sigma (\mathsf{c}) f\big)_ i(t) =  f_i(t),
 \end{equation*}
\item if $i = r \neq s:\ \forall t \in \mathbb{R}:$
\begin{equation*}
\big(b_\sigma (\mathsf{c})f\big)_{r}(t) = f_{s}\big(t-\xi\big)\chi_{[\xi,\, \xi+\sigma )}(t) + f_{r}(t) \chi_{[\xi,\, \xi+\sigma )^c}(t),
\end{equation*}
\item if $i = s \neq r:\ \forall t \in \mathbb{R}:$
\begin{equation*}
\big(b_\sigma (\mathsf{c})f\big)_{s}(t) = f_{r}\big(t + \xi\big)\chi_{[0,\sigma )}(t) + f_{s}(t)\chi_{[0,\sigma )^c}(t),
\end{equation*}
\item if $i = s= r:\ \forall t \in \mathbb{R}:$
\begin{equation*}\begin{split}
\big(b_\sigma& (\mathsf{c})f\big)_i(t) =  f_i\big(t-\xi\big)\chi_{[\xi,\, \xi+\sigma )}(t) + f_i\big(t + \xi\big)\chi_{[0,\sigma )}(t)\\
&+\,  f_i(t)\Big(\chi_{(-\infty,0)}(t) + \chi_{[\sigma , \xi)}(t)  + \chi_{[\xi+\sigma , \infty)}(t) \Big).
\end{split}\end{equation*}
\end{enumerate}
Note that $b_\sigma (\mathsf{c})$ is a unitary map on $L^2(\mathbb{R}; \mathbb{C}^m)$ and that $b_\sigma (\mathsf{c}_1)$ commutes with $b_\sigma (\mathsf{c}_2)$ for all $\mathsf{c}_1, \mathsf{c}_2 \in \mathcal{C}$. We 
define the \textbf{one-particle braid map} on $L^2(\mathbb{R}; \mathbb{C}^m)$ as
\begin{equation*}
b_\sigma  = \prod_{\mathsf{c}\in \mathcal{C}} b_\sigma  (\mathsf{c}) .
\end{equation*}
Furthermore, we may include the $\sigma =0$ case by taking $b_0$ to be the identity map on $ L^2(\mathbb{R})$. 
For $0\le \sigma \le \xi_{\text{min}}$, let $B_\sigma : \mathcal{F} \to \mathcal{F}$ be the second-quantisation of $b_\sigma $, and this defines a unitary. We call $B_\sigma $ the \textbf{braid map}. We extend $B_\sigma $ to $\mathcal{H}\otimes \mathcal{F}$, by ampliation with the identity on $\mathcal{H}$.
\end{definition}

\subsection{Time evolution of the network}
For $z \in \mathbb{R}$, we shall denote by $\lfloor z\rfloor$ 
the floor of $z$, i.e., the largest integer less
than or equal to $z$.

\begin{definition}{\bf \cite{BoG26}}\label{definition time evolution network}
Let $0 < \sigma \le \xi_{\text{min}}$. We define a unitary cocycle $\mathcal{U}_t$  on 
$\mathcal{H}\otimes \mathcal{F}$ by
\begin{equation*}
\mathcal{U}_t = \Theta_{\lfloor \tfrac{t}{\sigma } \rfloor \sigma }^* B_{t- \lfloor \tfrac{t}{\sigma } \rfloor \sigma }  U_{t- \lfloor \tfrac{t}{\sigma } \rfloor \sigma } \Theta_{\lfloor \tfrac{t}{\sigma } \rfloor \sigma }\overleftarrow{\prod_{i = 0}^{\lfloor \tfrac{t}{\sigma } \rfloor-1}}\Theta_{i\sigma }^* B_\sigma   U_\sigma  \Theta_{i\sigma }, 
\end{equation*}
where the product is the identity $I$ if $\lfloor \tfrac{t}{\sigma } \rfloor <1$.
Note that the definition of $\mathcal{U}_t$ does not depend 
on which $\sigma  \in (0,\xi_{\text{min}}]$ has been used.
\end{definition}
Note that $\mathcal{U}_t$ is the cocycle that governs the 
time evolution of the connected network. The physical time evolution 
of the network is given by the following one-parameter group of unitaries
\begin{equation}
\label{definition CalUhat}
\hat{\mathcal{U}}_t =  \left\{ \begin{array}{ll}
 \Theta_t \mathcal{U}_t & \mbox{\ \ \ if \ } t \ge 0 \\
 \mathcal{U}^*_{-t} \Theta_t & \mbox{\ \ \ if \ } t < 0 
  \end{array}\right. .
\end{equation}

\section{A cascaded system}

Let $N=2$, $m_1 = m_2 \in \mathbb{N}\backslash\{0\}$
and $m= 2m_1=2m_2$.
Suppose $(S^1, L^1, H^1)$ and $(S^2, L^2, H^2)$ are 
the $SLH$-triples of component 1 and 2, respectively. 
Let $(S, L, H)$ denote the $SLH$-triple of the network 
that consists of the two components 
(see Definition \ref{definition SLH concatenation})
\begin{equation*}
(S, L, H) = (S^1, L^1, H^1) \boxplus (S^2, L^2, H^2).
\end{equation*}
For  $\tau > 0$, 
define the following  
set of connections for the network given by $(S,L,H)$
\begin{equation}\label{equation connection set}
\mathcal{C} = \Big\{\left(1,m_1+1,\tau\right), \left(2,m_1+2,\tau\right),\dots,\left(m_1, m,\tau\right)\Big\}.
\end{equation}
Let the network  evolution $\mathcal{U}^\tau_t$ be 
given by Definition \ref{definition time evolution network}
with $\sigma =\tau$, $U_t$ 
given by equation \eqref{definition Ut}, and 
the connection set $\mathcal{C}$ given by 
equation \eqref{equation connection set}.

Define a unitary process $P_t$ by 
\begin{widetext}
\begin{equation*}
dP_t =\sum_{i = 1}^{\frac{m}{2}}  \left\{d\Lambda^{i, \frac{m}{2}+ i}_t + d\Lambda^{\frac{m}{2}+ i, i}_t
 - d\Lambda^{i,i}_t - d\Lambda^{\frac{m}{2} + i, \frac{m}{2} + i}_t\right\}P_t, \qquad P_0 = I.
\end{equation*}
Note that $P_t$ swaps the channels of component $1$ and $2$ on the 
Fock space $\mathcal{F}_{[0,t]}$ and acts 
as the identity on $\mathcal{F}_{(-\infty,0)}$ and 
$\mathcal{F}_{(t,\infty)}$. Let $\Theta_t^1$ and 
$\Theta^2_t$ be the left shift on the channels of 
component $1$ and the channels of component $2$, 
respectively. 
Define $V_t := U_t^2 P_t U_t^1$, then 
$V_t$ is a cocycle with respect to the shift $\Theta_t$.

\begin{theorem}
The network evolution $\mathcal{U}_t^\tau$ is given by 
\begin{equation}\label{equation second result}\begin{split}
\mathcal{U}^\tau_t { } &=  \Big(\Theta_\tau^{2*} P_t U^1_t\Theta^2_\tau\Big)U^2_t,\qquad\qquad\qquad (0 \le t \le \tau),\\
\mathcal{U}^\tau_t { } &=   
\Theta^1_\tau\Big(\Theta_t^* P_\tau U^1_\tau\Theta_t\Big)\Big( \Theta_{\tau}^*V_{t-\tau} \Theta_{\tau}\Big)\Theta^{1*}_\tau U^2_\tau, \qquad\qquad\qquad (t \ge \tau).
\end{split}\end{equation}
\end{theorem}
\begin{proof}
It follows from 
Definition \ref{definition braid map} with
the set of connections $\mathcal{C}$ defined in equation \eqref{equation connection set}, 
that: 
\begin{equation}\label{equation braid explicit}
B_t = \Theta_\tau^{2*} P_t \Theta^2_\tau,\qquad\qquad (0 \le t\le \tau).
\end{equation}
For $0\le t \le \tau$, Definition \ref{definition time evolution network} leads to
\begin{equation*}
\mathcal{U}_t^\tau = B_t   U_t  =\Theta_\tau^{2*} P_t \Theta^2_\tau U^1_tU^2_t =
\Big(\Theta_\tau^{2*} P_t U^1_t\Theta^2_\tau\Big)U^2_t.
\end{equation*}
Introduce the shorthand $\Delta =t- \lfloor \tfrac{t}{\tau } \rfloor \tau$. For $t \ge \tau$, Definition \ref{definition time evolution network} leads to
\begin{equation*}\begin{split}
\mathcal{U}^\tau_t { } &= \Theta_{\lfloor \tfrac{t}{\tau } \rfloor \tau }^* 
B_\Delta  U_\Delta \Theta_{\lfloor \tfrac{t}{\tau } \rfloor \tau }\overleftarrow{\prod_{i = 0}^{\lfloor \tfrac{t}{\tau } \rfloor-1}}\Theta_{i\tau }^* B_\tau   U_\tau  \Theta_{i\tau } 
=
 \Theta_{\lfloor \tfrac{t}{\tau } \rfloor \tau }^* 
B_\Delta U_\Delta \Theta_\tau 
\Big( B_\tau   U_\tau  \Theta_{\tau }\Big)^{\lfloor \tfrac{t}{\tau } \rfloor-1} B_\tau U_\tau \\
{ } &=  \Theta_{\lfloor \tfrac{t}{\tau } \rfloor \tau }^* B_\Delta U_\Delta \Theta_\tau 
\Big(
\Theta_\tau^{2*}P_\tau\Theta^2_\tau U^1_\tau U^2_\tau\Theta_\tau
\Big)^{\lfloor \tfrac{t}{\tau } \rfloor-1} \Theta_\tau^{2*}P_\tau\Theta^2_\tau U^1_\tau U^2_\tau \\
{ } &=
\Theta_{\lfloor \tfrac{t}{\tau } \rfloor \tau }^* B_\Delta U_\Delta \Theta^1_\tau 
P_\tau\Theta^2_\tau U^1_\tau \Big(U^2_\tau\Theta_\tau^1
P_\tau\Theta^2_\tau U^1_\tau \Big)^{\lfloor \tfrac{t}{\tau } \rfloor-1}U^2_\tau\\
{ } &=
\Theta_{\lfloor \tfrac{t}{\tau } \rfloor \tau }^* B_\Delta U_\Delta \Theta^1_\tau 
P_\tau U^1_\tau\Theta^2_\tau  \Big(\Theta_\tau^1U^2_\tau
P_\tau U^1_\tau\Theta^2_\tau  \Big)^{\lfloor \tfrac{t}{\tau } \rfloor-1}U^2_\tau\\
{ } &=
\Theta_{\lfloor \tfrac{t}{\tau } \rfloor \tau }^* B_\Delta U_\Delta \Theta^1_\tau 
P_\tau U^1_\tau\Theta_\tau  
\Big(U^2_\tau
P_\tau U^1_\tau\Theta_\tau  \Big)^{\lfloor \tfrac{t}{\tau } \rfloor-1}\Theta^{1*}_\tau U^2_\tau\\
{ } &=
\Theta_{\lfloor \tfrac{t}{\tau } \rfloor \tau }^* B_\Delta U_\Delta \Theta^1_\tau 
P_\tau U^1_\tau\Theta_{\lfloor \tfrac{t}{\tau } \rfloor\tau}  
\overleftarrow{\prod_{i = 1}^{\lfloor \tfrac{t}{\tau } \rfloor-1}}
\Big(\Theta_{i\tau}^*V_\tau\Theta_{i\tau}  \Big)\Theta^{1*}_\tau U^2_\tau\\
{ } &=
\Theta_{\lfloor \tfrac{t}{\tau } \rfloor \tau }^* B_\Delta U_\Delta \Theta^1_\tau 
P_\tau U^1_\tau\Theta_{\lfloor \tfrac{t}{\tau } \rfloor\tau}  
\Big(\Theta_\tau^*V_{\left(\lfloor \tfrac{t}{\tau } \rfloor-1\right)\tau}\Theta_\tau  \Big)\Theta^{1*}_\tau U^2_\tau,
\end{split}\end{equation*}
where the last step follows from the cocycle property of $V_t$. Furthermore, we have
\begin{equation*}\begin{split}
&B_\Delta U_\Delta \Theta^1_\tau 
P_\tau U^1_\tau =
\Theta_\tau^{2*} P_\Delta \Theta^2_\tau 
U^1_\Delta U^2_\Delta \Theta^1_\tau P_\tau U^1_\tau = 
(\Theta_\tau^{2*} P_\Delta U^1_\Delta\Theta_\tau U^2_\Delta) 
\Theta^*_\Delta (P_{\tau -\Delta} U^1_{\tau - \Delta})\Theta_\Delta
(P_\Delta U^1_\Delta) = \\
& \Theta_\tau^{2*} \Theta_{\tau-\Delta}\Big(\Theta_{\tau-\Delta}^*(P_\Delta U^1_\Delta)\Theta_{\tau-\Delta}\Big) 
 (P_{\tau -\Delta} U^1_{\tau - \Delta})\Theta_\Delta
(U^2_\Delta P_\Delta U^1_\Delta) =
\Theta^1_\tau (\Theta^*_\Delta P_\tau U_\tau^1\Theta_\Delta) V_\Delta.
\end{split}\end{equation*}
Substituting this into the expression for $\mathcal{U}^\tau_t$ yields for $t \ge \tau$
\begin{equation*}\begin{split}
\mathcal{U}_t^\tau { } &= 
\Theta_{\lfloor \tfrac{t}{\tau } \rfloor \tau }^*
\Theta^1_\tau (\Theta^*_\Delta P_\tau U_\tau^1\Theta_\Delta) V_\Delta
\Theta_{\lfloor \tfrac{t}{\tau } \rfloor\tau}  
\Big(\Theta_\tau^*V_{\left(\lfloor \tfrac{t}{\tau } \rfloor-1\right)\tau}\Theta_\tau  \Big)\Theta^{1*}_\tau U^2_\tau\\
{ } &=
\Theta^1_\tau
 (\Theta^*_t P_\tau U_\tau^1\Theta_t)\Theta^*_{\lfloor \tfrac{t}{\tau }\rfloor\tau} V_\Delta
\Theta_{\lfloor \tfrac{t}{\tau } \rfloor\tau}  
\Big(\Theta_\tau^*V_{\left(\lfloor \tfrac{t}{\tau } \rfloor-1\right)\tau}\Theta_\tau  \Big)\Theta^{1*}_\tau U^2_\tau\\
{ } &=
\Theta^1_\tau
 (\Theta^*_t P_\tau U_\tau^1\Theta_t)
(\Theta_\tau^*V_{t-\tau}\Theta_\tau)  \Theta^{1*}_\tau U^2_\tau. 
\end{split}\end{equation*}
\end{proof}

\section{The series product}

Since the systems are cascaded, it is natural 
to keep the same name for a channel after a 
connection was made. Therefore, one usually 
transports the output back to the channels of 
system $1$. To this end, 
define $\mathcal{V}^\tau_t := P_{t+\tau} \mathcal{U}^\tau_t$ for $t \ge \tau$, 
$\mathcal{V}^\tau_t := \Theta_\tau^* P_t\Theta_\tau P_t \mathcal{U}^\tau_t$ for $0 \le t < \tau$,
and $\mathcal{V}_t := P_tV_t$.

\begin{theorem}
The finite delay time evolution $\mathcal{V}^\tau_t$ converges 
strongly, uniformly on compact time intervals, to $\mathcal{V}_t$:
for all $\psi \in \mathcal{H}\otimes \mathcal{F}$ and 
all $T < \infty$:
\begin{equation}\label{equation main result}
 \sup_{t \in [0,T]}\left\| \big(\mathcal{V}^\tau_t - \mathcal{V}_t\big)\psi\right\|
\xrightarrow{\tau\to0} 0.
\end{equation}
Furthermore, the cocycle $\mathcal{V}_t$ satisfies 
the Hudson-Parthasarathy equation \eqref{definition Ut} with 
$SLH$-triple:
\begin{equation}\begin{split}\label{equation series product}
&(S,L,H) = \left(
\begin{bmatrix}
\tilde{S}^2S^1 & 0\\
0 & I
\end{bmatrix},\ 
\begin{bmatrix}
\tilde{L}^2 +\tilde{S}^2L^1 \\
0
\end{bmatrix},\ H^1+ H^2 + \mbox{Im}\left((\tilde{L}^{2*})^T \tilde{S}^2L^1\right)\right), \\
&\mbox{where:}\ \ \tilde{S}^2_{i,j} = S^2_{m_1+i, m_1+j},\ \ \tilde{L}^2_i = L^2_{m_1+i},\qquad (1 \le i,j \le m_1)
\end{split}\end{equation}
which is precisely the series product defined in \cite{GoJ09, GoJ09b, Gar93}.
\end{theorem}
 
\begin{proof}
Note that $\mathcal{V}^\tau_0 =\mathcal{V}_0$. Since the cocycles $P_t, U_t^1, U_t^2, V_t$ and 
the shifts $\Theta^1_t$ and $\Theta_t$ are strongly continuous in $t$, 
we find for all $\psi \in \mathcal{H}\otimes \mathcal{F}$ and $t>0$:
\begin{equation*}
\lim_{\tau \to 0}\mathcal{V}^\tau_t\psi 
=
 \lim_{\tau \to 0}P_{t+\tau}\Theta^1_\tau\Big(\Theta_t^* P_\tau U^1_\tau\Theta_t\Big)\Big( \Theta_{\tau}^*V_{t-\tau} \Theta_{\tau}\Big)\Theta^{1*}_\tau U^2_\tau\psi 
=P_tV_t\psi = \mathcal{V}_t\psi .
\end{equation*}
Note that $\hat{\mathcal{V}}^\tau_t$ and $\hat{\mathcal{V}}_t$
(see equation \eqref{definition CalUhat} for the definition)
are both strongly continuous unitary one-parameter groups on 
$\mathcal{H}\otimes\mathcal{F}$. Since the shift $\Theta_t$ is 
unitary, we have for all $\psi \in \mathcal{H}\otimes\mathcal{F}$
\begin{equation*}
\left\| \left(\hat{\mathcal{V}}^\tau_t - \hat{\mathcal{V}}_t\right) \psi \right\| = \left\| \big(\mathcal{V}^\tau_t - \mathcal{V}_t\big)\psi\right\|
\xrightarrow{\tau\to0} 0.
\end{equation*}
Due to the Trotter-Kato Theorem, see e.g.\ \cite[Ch.~1 Thm 6.1 (b) $\Rightarrow$ (a)]{KuE86}, we can now strengthen 
the convergence of $\hat{\mathcal{V}}^\tau_t$ to $\hat{\mathcal{V}}_t$ to strong convergence uniformly on compact 
time intervals, i.e.\
for all $\psi \in \mathcal{H}\otimes \mathcal{F}$ and 
all $T < \infty$:
\begin{equation*}
\sup_{t \in [0,T]}\left\| \left(\hat{\mathcal{V}}^\tau_t - \hat{\mathcal{V}}_t\right) \psi \right\|
\xrightarrow{\tau\to0} 0.
\end{equation*} 
Since the shift is unitary (see the definition in equation \eqref{definition CalUhat}), 
strong convergence, uniformly on compact time intervals, of 
the cocyles follows as well. That is, the result in equation \eqref{equation main result}
has been established.

Define $\tilde{U}_t^2 := P_tU_t^2 P_t = P_t^*U_t^2 P_t$. Note that sandwiching with $P_t$ swaps the 
channels of components $1$ and $2$, i.e.\  $\tilde{U}_t^2$ satisfies:
\begin{equation*}\begin{split}
&d\tilde{U}_t^2 = \Bigg\{\sum_{i,j = 1}^{m_1}\big(\tilde{S}^2_{i j}-\delta_{ij}\big) d \Lambda^{ij}_t + 
      \sum_{i=1}^{m_1} \tilde{L}^2_i dA^{i*}_{t} 
      -\sum_{i,j = 1}^{m_1} \tilde{L}^{2*}_i \tilde{S}^2_{i j}dA^j_t 
    -\frac{1}{2}\sum_{i= 1}^{m_1} \tilde{L}^{2*}_i\tilde{L}^2_i dt - iH^2dt\Bigg\}\tilde{U}_t^2, \qquad \tilde{U}^2_0 = I, \\
&\tilde{S}^2_{i,j} = S^2_{m_1+i, m_1+j},\ \ \tilde{L}^2_i = L^2_{m_1+i},\qquad (1 \le i,j \le m_1).
\end{split}\end{equation*}
Note that this can also be seen by 
writing out the product $P_t U^2_t P_t $
using the integration by parts rule, 
equation \eqref{equation integration by parts}, and 
the quantum It\^o table.

Next, writing out $\mathcal{V}_t = P_t U^2_t P_t U^1_t = \tilde{U}^2_t U^1_t$
using equation \eqref{equation integration by parts} and 
the quantum It\^o table, yields
\begin{equation*}
d\mathcal{V}_t  = \Bigg\{\sum_{i,j = 1}^m\big(S_{i j}-\delta_{ij}\big) d \Lambda^{ij}_t + \sum_{i=1}^m L_i dA^{i*}_{t} 
      -\sum_{i,j = 1}^m L^*_i S_{i j}dA^j_t 
    -\frac{1}{2}\sum_{i= 1}^m L^*_iL_i dt - iHdt\Bigg\} \mathcal{V}_t,\qquad \mathcal{V}_0 = I,
\end{equation*}
with
\begin{equation*}
(S,L,H) = \left(
\begin{bmatrix}
\tilde{S}^2S^1 & 0\\
0 & I
\end{bmatrix},\ 
\begin{bmatrix}
\tilde{L}^2 +\tilde{S}^2L^1 \\
0
\end{bmatrix},\ H^1+ H^2 + \mbox{Im}\left((\tilde{L}^{2*})^T \tilde{S}^2L^1\right)\right).
\end{equation*}
\end{proof}
\end{widetext}
\bibliography{series_product}

\end{document}